\documentclass[11pt]{article}

\usepackage[margin=1.05in]{geometry}
\usepackage{amsmath,amssymb,amsthm,amsfonts}
\usepackage{mathtools}
\usepackage{array}
\usepackage{booktabs}
\usepackage{xcolor}
\definecolor{linkcol}{rgb}{0.1,0.1,0.45}
\usepackage[colorlinks=true,linkcolor=linkcol,citecolor=linkcol,urlcolor=linkcol]{hyperref}

\newtheorem{theorem}{Theorem}[section]
\newtheorem{proposition}[theorem]{Proposition}

\newtheorem{corollary}[theorem]{Corollary}
\theoremstyle{definition}
\newtheorem{definition}[theorem]{Definition}

\theoremstyle{remark}
\newtheorem{remark}[theorem]{Remark}

\title{Schur States and Many-Body Quantum Walks on Line Graphs}
\author{Musung Kang\\[2pt]
\small Department of Mathematical Sciences, Seoul National University,
Seoul, Republic of Korea\\
\small \texttt{musung098@snu.ac.kr}\\
\small ORCID \href{https://orcid.org/0009-0005-7509-4704}{0009-0005-7509-4704}}
\date{}

\newcommand{\Line}{\ell}
\newcommand{\CC}{\mathbb{C}}
\newcommand{\NN}{\mathbb{N}}
\newcommand{\cH}{\mathcal{H}}
\newcommand{\cF}{\mathcal{F}}
\newcommand{\cS}{\mathcal{S}}
\newcommand{\dd}{\,\mathrm{d}}
\newcommand{\vac}{\Omega}
\newcommand{\ket}[1]{|#1\rangle}
\newcommand{\bra}[1]{\langle#1|}
\newcommand{\bracket}[2]{\langle #1\mid #2\rangle}
\newcommand{\ketbra}[2]{|#1\rangle\!\langle#2|}
\newcommand{\Sym}{\operatorname{Sym}}

\newcommand{\Tr}{\operatorname{Tr}}
\newcommand{\Ran}{\operatorname{Ran}}
\newcommand{\per}{\operatorname{per}}

\begin{document}
\maketitle

\begin{abstract}
Let $H$ be a finite simple graph.  
A one-particle continuous-time quantum walk on its line graph $\Line H$ has one mode for each edge of $H$.  
This paper is a corrected and reorganized account of that setting.

We first record the matrix encoding of an edge-amplitude vector, 
its Schur state. 
The encoding is necessarily complex symmetric
if it is to be complex linear and to retain the global phase,
and it is normalized by $1/\sqrt{2}$.
It is unitary for the Frobenius inner product and is closed under a canonical tensor lift to the categorical product of graphs.
The tensor lift also pins down the convention.
Among all unimodular phase variants the symmetric choice is the only one that survives, and every relaxed variant reduces to a sign.

We then replace the labelled tensor-power description 
by the bosonic and fermionic Fock spaces 
over $\ell^2(EH)$, 
and assemble the standard free-particle apparatus in the form used afterwards. 
This covers the occupation-number basis, the second-quantized line-graph Hamiltonian,
permanent and determinant formulae for many-particle transition amplitudes,
and the one-particle reduced density matrix.
The one graph-theoretic input is the even-Eulerian criterion, 
which is the zero-sum $2$-flow criterion of Wang and Hu, 
recalled here with a short proof to fix conventions.  
Its translation through the incidence identity gives a real equimodular
full-support $-2$ eigenmode of $A(\Line H)$, 
and hence bosonic condensates of energy $-2N$ that are uniform over the edges of $H$.
The fermionic case is stated separately so that the Pauli constraint is explicit.
\end{abstract}

\section{Introduction}\label{sec:intro}

The vertices of the line graph $\Line H$ are the edges of a root graph $H$.
Thus a quantum walk on $\Line H$ describes a particle hopping between edge
modes of $H$, with a hop allowed when the corresponding root edges meet.
Line-graph Hamiltonians also provide standard examples of flat bands: the
unsigned incidence identity forces the spectrum of $A(\Line H)$ to lie above
$-2$, and the kernel of the incidence matrix is precisely the $-2$ eigenspace \cite{Mielke1991,Kollar2020}.

There are two distinct composition problems.
The categorical product of two root graphs gives a tensor product of two distinguishable one-particle systems. 
By contrast, $N$ identical particles on one line graph live in a symmetric or alternating tensor power. 
The latter construction, and hence Fock space, is the appropriate setting for many-body line-graph walks.

The purpose of this paper is expository and corrective, in three parts.
\begin{enumerate}
    \item We give a complex-linear Schur representation of one-particle edge
    amplitudes, resolving the distinction between a complex symmetric
    amplitude matrix and a Hermitian density matrix.  The representation is unitary.
    \item We isolate the tensor compatibility that actually survives at the
    one-particle level and prove that it determines the convention uniquely.
    Phase variants are gauges, only the symmetric choice survives the tensor
    lift, and every relaxed convention reduces to the signs $\pm1$
    (Theorem~\ref{prop:phi-uniqueness} and Remark~\ref{rem:phi-pm1}).
    We will use symmetric and exterior powers for identical particles.
    \item We set out the line-graph walk on bosonic and fermionic Fock space, with the standard free-particle formulae in the form we need, 
    and lift the even-Eulerian construction to many-body eigenstates.
\end{enumerate}
None of the Fock-space material is new. It is assembled here for the flat-band statement of Section~\ref{sec:eulerian} at the level of $N$-particle states.

Section~\ref{sec:oneparticle} fixes the one-particle notation and defines
Schur states.  Section~\ref{sec:tensor} proves the corrected tensor statement.
Section~\ref{sec:fock} develops the many-body Fock formulation, 
and Section~\ref{sec:reduced} explains which one-body data survive for a general many-body state. 
Section~\ref{sec:eulerian} recalls the even-Eulerian criterion and draws its Fock-space consequences.

\section{One-particle line-graph walks and Schur states}
\label{sec:oneparticle}

Throughout, $H=(VH,EH)$ is a finite simple graph with $n=|VH|$ and
$m=|EH|$.  Unless stated otherwise, $H$ is connected and $m>0$.  Its line
graph $\Line H$ has vertex set $EH$, and $e,f\in EH$ are adjacent in
$\Line H$ exactly when $e\ne f$ and $e\cap f\ne\varnothing$.

\subsection{The one-particle walk}

The one-particle Hilbert space is
\[
    \cH_H:=\ell^2(EH)\cong\CC^m,
\]
with orthonormal edge basis $\{\ket e:e\in EH\}$.  We use the one-particle
Hamiltonian and propagator
\[
    h_H:=A(\Line H),
    \qquad
    U_1(t):=e^{-i t h_H}.
\]
Thus $i\,\partial_t\ket{\psi(t)}=h_H\ket{\psi(t)}$ and
$\ket{\psi(t)}=U_1(t)\ket{\psi(0)}$.  A weighted Hermitian hopping matrix may
replace $A(\Line H)$ in Sections~\ref{sec:oneparticle}
to~\ref{sec:reduced}.  The incidence and flat-band statements in
Section~\ref{sec:eulerian} concern
the unweighted adjacency matrix.

\subsection{Complex symmetric Schur states}

Define the Schur space of $H$ by
\[
    \cS_H:=\left\{
      S\in\operatorname{Mat}_{n}(\CC):
      S^{\mathsf T}=S,\quad S_{vv}=0,\quad
      S_{uv}=0\text{ if }\{u,v\}\notin EH
    \right\}.
\]
It is a complex Hilbert space with the Frobenius inner product
\[
    \langle S,T\rangle_F:=\Tr(S^\dagger T).
\]

\begin{definition}[Schur representation and Schur state]
\label{def:schur-state}
For $e=\{u,v\}\in EH$, set
\[
    \iota_H\ket e:=\frac{1}{\sqrt2}
    \bigl(\ket u\!\bra v+\ket v\!\bra u\bigr),
\]
and extend $\iota_H:\cH_H\to\cS_H$ complex-linearly.  The Schur state of
$\ket\psi\in\cH_H$ at time $t$ is
\[
    S_\psi(t):=\iota_H U_1(t)\ket\psi.
\]
Equivalently, if $\psi_e(t)=\bracket{e}{U_1(t)\psi}$, then
\[
    [S_\psi(t)]_{uv}=
    \begin{cases}
      \psi_{\{u,v\}}(t)/\sqrt2,&\{u,v\}\in EH,\\
      0,&\text{otherwise.}
    \end{cases}
\]
\end{definition}

\begin{proposition}[Unitary Schur representation]
\label{prop:schur-unitary}
The map $\iota_H:\cH_H\to\cS_H$ is a complex-linear unitary isomorphism.
In particular,
\[
    \langle\iota_H\psi,\iota_H\varphi\rangle_F
      =\bracket{\psi}{\varphi},
    \qquad
    \|S_\psi(t)\|_F=\|\psi\|.
\]
\end{proposition}

\begin{proof}
For each edge $e=\{u,v\}$, the two nonzero entries of $\iota_H\ket e$ are
both $1/\sqrt2$.  Distinct edges have disjoint matrix support.  Therefore
\[
 \left\langle\iota_H\ket e,\iota_H\ket f\right\rangle_F=\delta_{ef}.
\]
These matrices form a basis of $\cS_H$, so $\iota_H$ is unitary.  The norm
identity follows because $U_1(t)$ is unitary.
\end{proof}

\begin{remark}[Why the amplitude matrix is not Hermitian]
\label{rem:symmetric-not-hermitian}
If a nonzero $S$ is Hermitian, 
then $i S$ is not Hermitian.  
Consequently a Hermitian amplitude encoding cannot obey 
$S_{a\psi+b\varphi}=aS_\psi+bS_\varphi$ 
for arbitrary complex $a,b$.
It also fails to retain a general global phase under time evolution. 
The Hermitian objects are instead the density operator
$\rho_\psi=\ketbra{\psi}{\psi}$ and its transported operator
$\iota_H\rho_\psi\iota_H^\dagger$ on $\cS_H$.
\end{remark}

The edge probabilities are recovered without losing their normalization.
Define the real symmetric matrix
\[
    W_\psi(t):=2\bigl(\overline{S_\psi(t)}\circ S_\psi(t)\bigr),
\]
where $\circ$ is the entrywise product.  Then, for $e=\{u,v\}$,
\[
    [W_\psi(t)]_{uv}=|\psi_e(t)|^2,
    \qquad
    \sum_{e\in EH}|\psi_e(t)|^2=1
\]
whenever $\psi$ is normalized.  The factor $2$ compensates for the
$1/\sqrt2$ in Definition~\ref{def:schur-state}.

\section{Tensor compatibility and uniqueness of the convention}
\label{sec:tensor}

For graphs $G$ and $K$, their categorical product $G\times K$ has vertex
set $VG\times VK$, with $(u,x)$ adjacent to $(v,y)$ exactly when
$\{u,v\}\in EG$ and $\{x,y\}\in EK$.

For $e=\{u,v\}\in EG$ and $f=\{x,y\}\in EK$, there are two lifted edges
in $G\times K$:
\[
 e\boxtimes_0 f:=\{(u,x),(v,y)\},
 \qquad
 e\boxtimes_1 f:=\{(u,y),(v,x)\}.
\]
Interchanging the names of the endpoints merely interchanges the two lifted
edges.  Define an isometry
\[
 \mathcal W_{G,K}:\cH_G\otimes\cH_K\longrightarrow\cH_{G\times K}
\]
on the edge basis by
\[
 \mathcal W_{G,K}(\ket e\otimes\ket f)
 :=\frac{1}{\sqrt2}
 \left(\ket{e\boxtimes_0 f}+\ket{e\boxtimes_1 f}\right).
\]
The two-edge sets associated with distinct pairs $(e,f)$ are disjoint, which
proves that $\mathcal W_{G,K}$ is indeed an isometry.  Its range is generally
a proper subspace of $\cH_{G\times K}$.

\begin{proposition}[Schur-tensor compatibility]
\label{prop:schur-tensor}
For matrices of compatible sizes,
\[
 (A\otimes B)\circ(C\otimes D)
 =(A\circ C)\otimes(B\circ D).
\]
Moreover, for $\psi\in\cH_G$ and $\varphi\in\cH_K$,
\[
 \iota_{G\times K}\mathcal W_{G,K}(\psi\otimes\varphi)
 =\iota_G(\psi)\otimes\iota_K(\varphi).
\]
Thus the Kronecker product of two normalized Schur states is a normalized
Schur state on the ambient graph $G\times K$.  Its corresponding edge state
lies in the paired-edge subspace $\Ran\mathcal W_{G,K}$.
\end{proposition}

\begin{proof}
Fix row index $(i,a)$ and column index $(j,b)$.  The corresponding entries on
both sides of the first identity equal $A_{ij}C_{ij}B_{ab}D_{ab}$.  For the
second identity, it is enough to
take $\psi=\ket e$ and $\varphi=\ket f$.  Each of the two lifted edges has
amplitude $1/\sqrt2$ under $\mathcal W_{G,K}$ and hence matrix entries $1/2$
under $\iota_{G\times K}$.  The corresponding entries satisfy
\[
 \bigl[\iota_G\ket e\otimes\iota_K\ket f\bigr]_{(i,a),(j,b)}
 =\frac1{\sqrt2}\frac1{\sqrt2}=\frac12,
\]
and every other entry vanishes.  Linearity gives the result for arbitrary
vectors.
\end{proof}

\begin{theorem}[Compatible with Tensor]
\label{prop:phi-uniqueness}
Let $\phi\in\CC$ with $|\phi|=1$.  Give every graph an orientation of its
edges, and for $e=\{u,v\}$ oriented from $u$ to $v$ put
\[
 \iota^\phi_H\ket e:=\frac1{\sqrt2}
 \bigl(\ket u\!\bra v+\phi\ket v\!\bra u\bigr),
\]
extended complex-linearly, with range
\[
 \cS^\phi_H:=\bigl\{S:\ S_{vu}=\phi\,S_{uv}\ \text{on oriented edges},\
 S_{uv}=0\ \text{if}\ \{u,v\}\notin EH\bigr\}.
\]
Then:
\begin{enumerate}
 \item[\emph{(i)}] For every $\phi$ the map $\iota^\phi_H:\cH_H\to\cS^\phi_H$
 is a complex-linear unitary, and the induced edge weights are the same for
 all $\phi$: writing $\iota^\phi_H=\Phi\circ\iota^1_H$ for the unimodular
 matrix $\Phi$ with $\Phi_{uv}=1$ and $\Phi_{vu}=\phi$ on oriented edges, one
 has $\overline{\iota^\phi_H\psi}\circ\iota^\phi_H\psi
 =\overline{\iota^1_H\psi}\circ\iota^1_H\psi$ for every $\psi$.
 \item[\emph{(ii)}] The tensor identity of
 Proposition~\ref{prop:schur-tensor},
 \[
  \iota^\phi_{G\times K}\mathcal W_{G,K}(\psi\otimes\varphi)
  =\iota^\phi_G(\psi)\otimes\iota^\phi_K(\varphi),
 \]
 holds for all $G$, $K$ and all orientations if and only if $\phi=1$.  The
 same condition characterizes the orientation-independent choice.
\end{enumerate}
\end{theorem}

\begin{proof}
(i) The two nonzero entries of $\iota^\phi_H\ket e$ have modulus $1/\sqrt2$
and distinct edges have disjoint matrix support, so
$\langle\iota^\phi_H\ket e,\iota^\phi_H\ket f\rangle_F=\delta_{ef}$ exactly as
in Proposition~\ref{prop:schur-unitary}.  The factorization
$\iota^\phi_H=\Phi\circ\iota^1_H$ is immediate from the definition, and
$|\Phi_{uv}|=1$ everywhere on $EH$, so the entrywise moduli, hence
$\overline S\circ S$, do not depend on $\phi$.

(ii) Take $\psi=\ket e$ with $e=\{u,v\}$ oriented $u\to v$ and
$\varphi=\ket f$ with $f=\{x,y\}$ oriented $x\to y$.  Expanding the Kronecker
product,
\[
 \iota^\phi_G\ket e\otimes\iota^\phi_K\ket f
 =\tfrac12\bigl(
   \ketbra{ux}{vy}
  +\phi\,\ketbra{uy}{vx}
  +\phi\,\ketbra{vx}{uy}
  +\phi^{2}\,\ketbra{vy}{ux}\bigr).
\]
On the lifted edge $e\boxtimes_0 f=\{(u,x),(v,y)\}$ the two entries are
$\tfrac12$ and $\tfrac12\phi^{2}$.  On $e\boxtimes_1 f=\{(u,y),(v,x)\}$ they
are $\tfrac12\phi$ and $\tfrac12\phi$.  Membership in
$\cS^\phi_{G\times K}$ requires the ratio of the reversed entry to the
forward entry to be $\phi$ or, if that lifted edge is oriented the other way,
$\phi^{-1}$.  On $e\boxtimes_1 f$ that ratio is $1$, so $\phi=1$ or
$\phi^{-1}=1$.  Since $|\phi|=1$, both give $\phi=1$.  Conversely, for
$\phi=1$ both ratios equal $1$ and the identity is
Proposition~\ref{prop:schur-tensor}.  Finally, reversing the orientation of a
single edge $e$ replaces $\iota^\phi_H\ket e$ by
$\tfrac1{\sqrt2}(\ket v\!\bra u+\phi\ket u\!\bra v)$, which agrees with
$\iota^\phi_H\ket e$ for every $e$ precisely when $\phi=1$.
\end{proof}

\begin{remark}
The choice $\phi=i$ illustrates both halves.  It is a complex-linear
unitary encoding with $(\iota^{i}_H\ket e)^\dagger=-i\,\iota^{i}_H\ket e$,
and by (i) it reproduces the edge probabilities $W_\psi$ exactly.  But the two
lifted edges above acquire ratios $-1$ and $1$, neither equal to $i$, so it
does not survive the tensor lift.  A phase inserted in the representation is a
gauge and is invisible to every quantity built from $|S_{uv}|$.  A phase that is
to change the dynamics belongs in the Hamiltonian, as $\Phi\circ A(\Line H)$.
\end{remark}

\begin{remark}[Relaxed conventions reduce to signs]
\label{rem:phi-pm1}
Theorem~\ref{prop:phi-uniqueness} fixes one convention for all
graphs at once.  It is natural to ask how much survives when the
convention is relaxed.  Let the left factor, the right factor, and the
ambient product graph carry independent phases $\phi_L$, $\phi_R$,
$\phi_A$, and let the lift be any isometry supported on the paired
edges.  The support of
$\iota^{\phi_L}_G\psi\otimes\iota^{\phi_R}_K\varphi$ pins such an
isometry to the two lifted edges, and matching the four matrix entries
shows that solutions exist exactly when $\phi_L^2=1$ or $\phi_R^2=1$,
with $\phi_A$ and the two lift amplitudes then determined by the
orientations.  Two cases stand out.  The choice
$(\phi_L,\phi_R,\phi_A)=(\phi,1,\phi)$ is the entrywise gauge of
part~(i) pushed to the left factor.  It breaks the symmetry between
the two factors.  The choice $(\phi_L,\phi_R,\phi_A)=(-1,-1,1)$
encodes both factors antisymmetrically and restores the factor
symmetry through the sign-twisted lift
\[
 \mathcal W^{-}_{G,K}(\ket e\otimes\ket f)
 =\tfrac1{\sqrt2}\bigl(\ket{e\boxtimes_0 f}-\ket{e\boxtimes_1 f}\bigr).
\]
Both cases depend on the chosen orientations.  Hence a convention that
treats the two factors equally and is to apply to the class of all
finite simple graphs is confined to $\phi\in\{\pm1\}$.  Demanding in
addition that the ambient product graph, itself a finite simple graph,
carries the same convention returns $\phi=1$ by part~(ii).
\end{remark}

For later comparison, note that the tensor product in
Proposition~\ref{prop:schur-tensor} describes two distinguishable factors.
Identical particles on one graph require the construction below.

\section{Many-body quantum walks on line graphs}\label{sec:fock}

Fix $H$ and abbreviate $\cH:=\cH_H$ and $h:=h_H$.  For $N\ge1$, let
$U_\pi$ denote the unitary that permutes the factors of $\cH^{\otimes N}$
according to $\pi\in\mathfrak S_N$, and define
\[
 P_+^{(N)}:=\frac1{N!}\sum_{\pi\in\mathfrak S_N}U_\pi,
 \qquad
 P_-^{(N)}:=\frac1{N!}\sum_{\pi\in\mathfrak S_N}
 \operatorname{sgn}(\pi)U_\pi.
\]
The $N$-boson and $N$-fermion spaces are
\[
 \cH_+^{(N)}:=\Sym^N\cH=P_+^{(N)}\cH^{\otimes N},
 \qquad
 \cH_-^{(N)}:=\bigwedge^N\cH=P_-^{(N)}\cH^{\otimes N}.
\]
The Fock spaces are
\[
 \cF_+(\cH):=\bigoplus_{N=0}^{\infty}\Sym^N\cH,
 \qquad
 \cF_-(\cH):=\bigoplus_{N=0}^{m}\bigwedge^N\cH,
\]
with vacuum sector $\Sym^0\cH=\bigwedge^0\cH=\CC\vac$.  Their fixed-sector
dimensions are
\[
 \dim\Sym^N\cH=\binom{m+N-1}{N},
 \qquad
 \dim\bigwedge^N\cH=\binom{m}{N}.
\]

\subsection{Occupation-number bases}

Fix an ordering $EH=\{e_1,\ldots,e_m\}$.  Bosonic creation and annihilation
operators satisfy
\[
 [b_e,b_f^\dagger]=\delta_{ef},
 \qquad [b_e,b_f]=[b_e^\dagger,b_f^\dagger]=0,
\]
whereas fermionic operators satisfy
\[
 \{c_e,c_f^\dagger\}=\delta_{ef},
 \qquad \{c_e,c_f\}=\{c_e^\dagger,c_f^\dagger\}=0.
\]
The normalized occupation bases are
\[
 \ket{\boldsymbol n}_+
 :=\prod_{j=1}^m\frac{(b_{e_j}^\dagger)^{n_j}}{\sqrt{n_j!}}\vac,
 \quad n_j\in\NN\cup\{0\},
 \quad \sum_jn_j=N,
\]
and
\[
 \ket{\boldsymbol n}_-
 :=(c_{e_1}^\dagger)^{n_1}\cdots(c_{e_m}^\dagger)^{n_m}\vac,
 \quad n_j\in\{0,1\},
 \quad \sum_jn_j=N.
\]
The displayed ordering fixes the fermionic sign convention.  We write
$a_{e,+}=b_e$ and $a_{e,-}=c_e$ when a statement applies to either
statistics.  The number operator is
\[
 \widehat N:=\sum_{e\in EH}a_{e,\pm}^\dagger a_{e,\pm}.
\]

\subsection{Second quantization}

The free many-body Hamiltonian is the second quantization of $h$:
\begin{align}
 \mathrm d\Gamma_\pm(h)\big|_{\cH_\pm^{(N)}}
 &:=\left.\sum_{r=1}^N
 I^{\otimes(r-1)}\otimes h\otimes I^{\otimes(N-r)}
 \right|_{\cH_\pm^{(N)}},                                      \label{eq:dGamma-sector}\\
 \mathrm d\Gamma_\pm(h)
 &=\sum_{e,f\in EH}\bra e h\ket f\,
 a_{e,\pm}^\dagger a_{f,\pm}.                                  \label{eq:dGamma-modes}
\end{align}
It annihilates the vacuum.  Since $h=A(\Line H)$,
\[
 \mathrm d\Gamma_\pm(h)=
 \sum_{\{e,f\}\in E(\Line H)}
 \left(a_{e,\pm}^\dagger a_{f,\pm}
 +a_{f,\pm}^\dagger a_{e,\pm}\right).
\]

A number-preserving two-body interaction may be added in the form
\[
 \widehat V_\pm=\frac12\sum_{e,f,g,k\in EH}
 V_{ef;gk}\,a_{e,\pm}^\dagger a_{f,\pm}^\dagger
 a_{k,\pm}a_{g,\pm},
\]
where the coefficients are the matrix elements of a self-adjoint two-body
operator that preserves the symmetric or alternating two-particle sector.
The interacting Hamiltonian is
$\widehat H_\pm=\mathrm d\Gamma_\pm(h)+\widehat V_\pm$.  Statements below
that explicitly use $\mathrm d\Gamma_\pm(h)$ concern the free model.  They do
not automatically survive an arbitrary interaction.

\begin{proposition}[Free many-body evolution]
\label{prop:free-evolution}
For either statistics and every $N$,
\[
 e^{-i t\mathrm d\Gamma_\pm(h)}\big|_{\cH_\pm^{(N)}}
 =\left.U_1(t)^{\otimes N}\right|_{\cH_\pm^{(N)}}.
\]
Moreover, $[\mathrm d\Gamma_\pm(h),\widehat N]=0$.
\end{proposition}

\begin{proof}
On $\cH^{\otimes N}$, the $N$ summands in
Equation~\eqref{eq:dGamma-sector} act on different factors and commute.
Exponentiating their sum therefore gives $U_1(t)^{\otimes N}$.  The sum is
invariant under factor permutations, so it restricts to both statistics.
Finally, the canonical commutation or anticommutation relations give
$[\widehat N,a_e^\dagger a_f]=0$ for every $e,f$, proving number
conservation.
\end{proof}

\subsection{Many-particle transition amplitudes}

For occupation vectors $\boldsymbol n$ and $\boldsymbol r$ with total number
$N$, let $U[\boldsymbol r\mid\boldsymbol n]$ be the $N\times N$ matrix formed
from $U=U_1(t)$ by repeating row $e$ exactly $r_e$ times and column $f$
exactly $n_f$ times.

\begin{proposition}[Permanent and determinant amplitudes]
\label{prop:per-det}
For bosons,
\[
 {}_+\!\bra{\boldsymbol r}
 e^{-i t\mathrm d\Gamma_+(h)}
 \ket{\boldsymbol n}_+
 =\frac{\per U[\boldsymbol r\mid\boldsymbol n]}
 {\sqrt{\prod_e r_e!\prod_f n_f!}}.
\]
For fermions, let $I=\{f:n_f=1\}$ and $J=\{e:r_e=1\}$, listed in the fixed
edge order.  Then
\[
 {}_-\!\bra{\boldsymbol r}
 e^{-i t\mathrm d\Gamma_-(h)}
 \ket{\boldsymbol n}_-
 =\det U[J,I].
\]
\end{proposition}

\begin{proof}
Second quantization gives the covariance relation
\[
 e^{-i t\mathrm d\Gamma_\pm(h)}a_f^\dagger
 e^{i t\mathrm d\Gamma_\pm(h)}
 =\sum_e U_{ef}a_e^\dagger.
\]
Apply it to each creation operator in the occupation vector.  In the bosonic
case, equal monomials are collected without signs, producing the permanent.
The factorials come from the normalized occupation basis.  In the fermionic
case, reordering each monomial into the fixed edge order contributes the
permutation sign, producing the determinant.
\end{proof}

\subsection{Schur states in Fock space}

The one-particle unitary $\iota_H$ induces, for either statistics,
\[
 \iota_{H,\pm}^{(N)}
 :=\left.\iota_H^{\otimes N}\right|_{\cH_\pm^{(N)}}:
 \cH_\pm^{(N)}\longrightarrow
 P_\pm^{(N)}\bigl(\cS_H^{\otimes N}\bigr).
\]

\begin{definition}[Many-body Schur state]
For $\Psi_N\in\cH_\pm^{(N)}$, its many-body Schur state is
\[
 \mathbb S_{\Psi_N}:=\iota_{H,\pm}^{(N)}\Psi_N.
\]
For a Fock vector $\Psi=\bigoplus_N\Psi_N$, set
\[
 \mathbb S_\Psi:=\bigoplus_N\mathbb S_{\Psi_N}.
\]
\end{definition}

\begin{proposition}[Schur-Fock compatibility]
\label{prop:schur-fock}
The map
\[
 \cF_\pm(\iota_H):=\bigoplus_N\iota_{H,\pm}^{(N)}
\]
is unitary onto its Schur-Fock image.  It intertwines the free Hamiltonians:
\[
 \cF_\pm(\iota_H)\,\mathrm d\Gamma_\pm(h)\,
 \cF_\pm(\iota_H)^\dagger
 =\mathrm d\Gamma_\pm(\iota_Hh\iota_H^\dagger).
\]
It also preserves symmetrized and alternating products.
\end{proposition}

\begin{proof}
Every tensor power of the unitary $\iota_H$ is unitary.  Since the same map
is applied to every factor, $\iota_H^{\otimes N}$ commutes with all factor
permutations and hence with $P_\pm^{(N)}$.  It therefore restricts to the
claimed unitary on each sector.  Conjugating each summand in
Equation~\eqref{eq:dGamma-sector} proves the intertwining identity.  The same
commutation with $P_\pm^{(N)}$ proves compatibility with symmetrized or
alternating products.
\end{proof}

For two independent mode systems, the natural one-particle operation is a
direct sum, not a categorical graph product.  The fixed-number decompositions
are
\[
 \Sym^N(\cH_1\oplus\cH_2)
 \cong\bigoplus_{p+q=N}\Sym^p\cH_1\otimes\Sym^q\cH_2,
\]
and
\[
 \bigwedge^N(\cH_1\oplus\cH_2)
 \cong\bigoplus_{p+q=N}\bigwedge^p\cH_1\otimes\bigwedge^q\cH_2.
\]
This is the Fock-space replacement for treating identical particles as an
unrestricted ordinary tensor product.

\section{One-particle reductions of many-body states}\label{sec:reduced}

A general many-body pure state is not determined by a single $n\times n$
Schur amplitude matrix.  Its canonical one-body datum is the reduced density
matrix.

\begin{definition}[One-particle reduced density matrix]
Let $\Psi_N\in\cH_\pm^{(N)}$ be normalized, with $N\ge1$.  Define
\[
 [\gamma_{\Psi_N}]_{ef}
 :=\bra{\Psi_N}a_f^\dagger a_e\ket{\Psi_N}.
\]
Then $\gamma_{\Psi_N}\succeq0$ and $\Tr\gamma_{\Psi_N}=N$.  The normalized
one-body state is $\rho_{\Psi_N}^{(1)}:=\gamma_{\Psi_N}/N$.
\end{definition}

Positivity follows from
\[
 \bra z\gamma_{\Psi_N}\ket z
 =\left\|\sum_e \overline{z_e}a_e\Psi_N\right\|^2\ge0,
\]
and the trace identity is $\bracket{\Psi_N}{\widehat N\Psi_N}=N$.
The normalized edge occupation distribution is
\[
 p_e(t):=\frac1N[\gamma_{\Psi_N}(t)]_{ee},
 \qquad \sum_{e\in EH}p_e(t)=1.
\]
It induces a real symmetric weighted adjacency matrix on the root graph by
assigning weight $p_{\{u,v\}}(t)$ to $\{u,v\}\in EH$.

\begin{proposition}[Free evolution of the one-body reduction]
\label{prop:1rdm-evolution}
Under the free Hamiltonian $\mathrm d\Gamma_\pm(h)$,
\[
 \gamma_{\Psi_N}(t)=U_1(t)\gamma_{\Psi_N}(0)U_1(t)^\dagger.
\]
If $h=\sum_r\theta_rE_r$ is its spectral decomposition over distinct
eigenvalues, then
\[
 \lim_{T\to\infty}\frac1T\int_0^T\gamma_{\Psi_N}(t)\dd t
 =\sum_rE_r\gamma_{\Psi_N}(0)E_r.
\]
\end{proposition}

\begin{proof}
The creation-annihilation covariance used in
Proposition~\ref{prop:per-det} gives the first identity entrywise.  Expanding
$\gamma(0)=\sum_{r,s}E_r\gamma(0)E_s$ gives oscillatory factors
$e^{-i t(\theta_r-\theta_s)}$.  Their Ces\`aro means vanish for $r\ne s$
and equal $1$ for $r=s$, proving the second identity.
\end{proof}

Two model classes make the reduction explicit.  An \emph{orbital} is a
normalized one-particle state $\phi\in\cH_H$.  The \emph{bosonic
condensate} in $\phi$ is the $N$-fold occupation of this single
orbital,
\[
 \ket{N:\phi}:=\frac{b^\dagger(\phi)^N}{\sqrt{N!}}\vac,
 \qquad
 b^\dagger(\phi):=\sum_{e\in EH}\bracket{e}{\phi}\,b_e^\dagger.
\]
From $b_e\ket{N:\phi}=\sqrt N\,\bracket{e}{\phi}\ket{N-1:\phi}$ one
computes
\[
 [\gamma^{(1)}]_{ef}
 =\bra{N:\phi}b_f^\dagger b_e\ket{N:\phi}
 =N\bracket{e}{\phi}\overline{\bracket{f}{\phi}},
 \qquad\text{that is,}\qquad
 \gamma^{(1)}=N\ketbra{\phi}{\phi}.
\]
Physically this is complete Bose--Einstein condensation in the sense of
Penrose and Onsager~\cite{PenroseOnsager1956}.  The one-body reduction
has a single macroscopic eigenvalue $N$, and its eigenvector $\phi$ is
the condensate orbital.  For a fermionic Slater determinant built from
orthonormal orbitals $\phi_1,\ldots,\phi_N$,
\[
 \gamma^{(1)}=\sum_{j=1}^N\ketbra{\phi_j}{\phi_j},
\]
a rank $N$ projector, so no eigenvalue exceeds $1$.
These formulae also show why the one-particle Schur matrix fully describes a
condensate orbital but not a general correlated many-body state.

\section{Even-Eulerian flat bands and their Fock lift}
\label{sec:eulerian}

Let $B\in\{0,1\}^{VH\times EH}$ be the unsigned incidence matrix,
\[
 B_{v,e}=\begin{cases}1,&v\in e,\\0,&v\notin e.\end{cases}
\]
For a simple graph,
\begin{equation}\label{eq:incidence-line}
 B^{\mathsf T}B=2I_m+A(\Line H),
 \qquad
 \ker B=\ker\bigl(A(\Line H)+2I_m\bigr).
\end{equation}
Indeed, a diagonal entry of $B^{\mathsf T}B$ equals $2$, while an
off-diagonal entry equals $1$ precisely when the two distinct root edges meet.

The combinatorial equivalence (1)$\Leftrightarrow$(2) below is exactly the
zero-sum $2$-flow criterion of Wang and Hu~\cite[Lemma~3]{WangHu2012}, where
both the componentwise statement and the Euler-tour proof appear.  The
zero-sum flow framework goes back to~\cite{Akbari2009}.  We recall it with a
short self-contained proof to fix conventions.  What this paper uses is not
the criterion itself but its translation, through the incidence
identity~\eqref{eq:incidence-line}, into a real equimodular full-support
flat-band orbital of $A(\Line H)$.

\begin{proposition}[Even-Eulerian criterion and its flat-band orbital]
\label{prop:even-eulerian}
Let $H=(V,E)$ be a finite connected simple graph with $m:=|E|>0$.  The
following conditions are equivalent.
\begin{enumerate}
 \item Every vertex of $H$ has even degree and $m$ is even.
 \item There is a signing $s:E\to\{+1,-1\}$ such that
 \[
     \sum_{e\ni v}s(e)=0
     \qquad\text{for every }v\in V.
 \]
\end{enumerate}
When these conditions hold,
\[
 \ket{\phi_s}:=\frac1{\sqrt m}\sum_{e\in E}s(e)\ket e
\]
is a uniform full-support eigenvector of $A(\Line H)$ with eigenvalue $-2$.
In particular,
\[
 |\bracket{e}{\phi_s}|^2=\frac1m
 \quad(e\in E),
 \qquad
 [A(\Line H),\ketbra{\phi_s}{\phi_s}]=0.
\]
No regularity assumption on $\Line H$ is required.
\end{proposition}

\begin{proof}
Assume first that every vertex has even degree and $m$ is even.  Since $H$ is
connected, Euler's theorem~\cite{Bondy2008} gives an Euler circuit.  Write it
with cyclic indices as
\[
 v_0,e_1,v_1,e_2,\ldots,e_m,v_m,
 \qquad v_m=v_0,
 \qquad e_j=\{v_{j-1},v_j\}.
\]
Set $s(e_j)=(-1)^{j-1}$ and $e_{m+1}=e_1$.  At every occurrence of a vertex,
the incoming incidence from $e_j$ is paired with the outgoing incidence from
$e_{j+1}$.  These pairs partition all incidences at every vertex.  For
$1\le j<m$ their signs sum to
$(-1)^{j-1}+(-1)^j=0$.  The closing pair also sums to zero because
\[
 s(e_m)+s(e_1)=(-1)^{m-1}+1=0
\]
when $m$ is even.  Hence $Bs=0$.

Conversely, suppose $s:E\to\{+1,-1\}$ and $Bs=0$.  At each vertex the number
of positive incident edges equals the number of negative incident edges, so
every degree is even.  Since every edge has two endpoints,
\[
 0=\boldsymbol{1}_V^{\mathsf T}Bs
   =2\boldsymbol{1}_E^{\mathsf T}s
   =2\sum_{e\in E}s(e).
\]
Thus the positive and negative edge counts are equal, and $m$ is even.

Finally, Equation~\eqref{eq:incidence-line} gives
\[
 A(\Line H)s=(B^{\mathsf T}B-2I_m)s=-2s.
\]
Normalization yields $\phi_s$ and its uniform edge probabilities.  Since
$A(\Line H)$ is real symmetric,
\[
 A(\Line H)\ketbra{\phi_s}{\phi_s}
 =-2\ketbra{\phi_s}{\phi_s}
 =\ketbra{\phi_s}{\phi_s}A(\Line H),
\]
which proves the commutator statement.
\end{proof}

The incidence identity also has a useful Fock form.  Set
\[
 d_{v,\pm}:=\sum_{e\ni v}a_{e,\pm}.
\]
These are not independent canonical modes, but direct expansion gives
\begin{equation}\label{eq:incidence-fock}
 \mathrm d\Gamma_\pm(A(\Line H))
 =\sum_{v\in VH}d_{v,\pm}^\dagger d_{v,\pm}-2\widehat N.
\end{equation}
Thus the $-2$ flat band is the one-particle kernel of all incidence-mode
annihilators.

\begin{corollary}[Bosonic flat-band condensates]
\label{cor:boson-flat}
Under the hypotheses of Proposition~\ref{prop:even-eulerian}, form the
bosonic condensate of Section~\ref{sec:reduced} in the orbital
$\phi_s$,
\[
 \ket{N:\phi_s}=\frac{b^\dagger(\phi_s)^N}{\sqrt{N!}}\vac,
 \qquad
 b^\dagger(\phi_s)=\sum_{e\in E}\bracket{e}{\phi_s}\,b_e^\dagger.
\]
For every $N\ge0$,
\[
 \mathrm d\Gamma_+(A(\Line H))\ket{N:\phi_s}
 =-2N\ket{N:\phi_s},
\]
and the edge occupation is uniform:
\[
 \bra{N:\phi_s}b_e^\dagger b_e\ket{N:\phi_s}=\frac Nm
 \qquad(e\in E).
\]
\end{corollary}

\begin{proof}
The state is the symmetric tensor $\phi_s^{\otimes N}$, so its energy is the
sum of $N$ copies of the eigenvalue $-2$.  Also
$b_e\ket{N:\phi_s}=\sqrt N\,\bracket{e}{\phi_s}\ket{N-1:\phi_s}$,
which gives the occupation formula.
\end{proof}

For fermions, the same orbital cannot be occupied twice.  More generally, if
$\phi_1,\ldots,\phi_N$ are orthonormal vectors in $\ker B$, then the Slater
state
\[
 c^\dagger(\phi_1)\cdots c^\dagger(\phi_N)\vac
\]
has free energy $-2N$.  Such a state exists only when
$N\le\dim\ker B$.  Its occupation at edge $e$ is
$\sum_{j=1}^N|\bracket{e}{\phi_j}|^2$.  It is uniform only under the additional
condition that the projector onto
$\operatorname{span}\{\phi_1,\ldots,\phi_N\}$ has constant diagonal.  The
even-Eulerian theorem alone does not imply this extra condition.

\section{Concluding remarks}\label{sec:conclusion}

The complex symmetric Schur representation separates three objects that
should not be conflated: an edge-amplitude vector, its supported symmetric
matrix, and its Hermitian density operator.  The corrected categorical tensor
lift is exact for distinguishable one-particle systems, whereas symmetric and
exterior powers give the many-body state spaces for identical particles.
Within Fock space, the line-graph Hamiltonian has the standard quadratic form,
its transition amplitudes are permanents or determinants, and its one-body
reduction evolves by the original line-graph walk.  The tensor lift fixes
the representation uniquely.  Any convention that applies to the class of
all finite simple graphs is confined to the signs $\pm1$, and the symmetric
choice is the only orientation-free one.

The even-Eulerian theorem provides a distinguished uniform flat-band
orbital.  It yields a uniform bosonic condensate at every particle number.
Fermionic flat-band states require several orthogonal incidence-kernel modes,
and uniform density is an additional projector condition.  For interacting
models, determining which incidence-kernel sectors remain invariant is a
separate problem.  No such invariance is assumed in this paper.

\subsection*{Acknowledgements}
The author is grateful for the Hyunsong Graduate Scholarship.


\end{document}